\documentclass[a4paper,11pt]{article}

\usepackage{inputenc}
\usepackage{float}
\usepackage[margin=0.62in]{geometry}
\usepackage{natbib}
\setcitestyle{numbers,open={[},close={]}}
\usepackage{amsmath}	
\usepackage{amssymb}	
\usepackage{setspace}
\usepackage{relsize, wrapfig}

\usepackage{authblk}
\usepackage{graphics}
\usepackage{xcolor}

\usepackage{amsmath, amssymb, amsthm}
\usepackage{graphicx}

\usepackage{adjustbox}

\usepackage[normalem]{ulem}
\useunder{\uline}{\ul}{}

\usepackage{bm}

\newtheorem{theorem}{Theorem}[section]

\newtheorem{prop}[theorem]{Proposition}

\begin{document}

\title{Imputation under Differential Privacy}
\author[1]{Soumojit Das}
\author[1, 2]{J\"{o}rg Drechsler}
\author[3]{Keith Merrill}
\author[1]{Shawn Merrill}
\affil[1]{The Joint Program in Survey Methodology, University of Maryland, College Park}
\affil[2]{Institute for Employment Research}
\affil[3]{Brandeis University}
\renewcommand\Affilfont{\itshape\small}

\def\GS{\operatorname{GS}}
\def\datauni{\mathcal{U}}
\renewcommand{\Pr}{\mathbb{P}}

\date{}

\maketitle

\section{Introduction}
The literature on differential privacy almost invariably assumes that the data to be analyzed are fully observed. In most practical applications this is an unrealistic assumption. A simple strategy that is commonly applied to deal with this problem is to discard all those units that are not fully observed on those attributes used in the analysis (available-case analysis). While convenient, this approach will typically be inefficient, as information that is still partially available is not used. However, the real concern is that the obtained results will be biased in most cases if the distribution of the fully observed data differs from the distribution of the entire data. More formally, available-case analysis tends to be biased if the probability for an attribute to be missing is correlated with the information contained in the data.

A popular strategy to address this problem is imputation. With imputation, missing values are replaced by estimated values given the observed data. 
While alternative strategies 
exist to properly take the missingness into account, the simplicity of the imputation approach makes it a convenient tool 
that is commonly used in many applied fields as a data preprocessing step before analyzing the data. 

If privacy is a concern, the natural question arises how the imputation step affects the guarantees of formal privacy methods such as differential privacy. 
This paper aims to initiate the research regarding the interplay between differential privacy and imputation, 
offering the following contributions:
\begin{itemize}
    \item Borrowing ideas from the notion of group privacy, we show that na\"ively running a DP mechanism on the imputed data can lead to privacy degradation.
    \item We illustrate that in the worst case, the sensitivity of the query increases linearly with the number of missing data points if privacy is only taken into account when running the query of interest.
    \item We demonstrate that for a general class of imputation strategies, the worst case bounds can be improved by ensuring privacy already at the imputation stage. 
\end{itemize}



\subsection{Related work}
The literature on imputation in the context of differential privacy is surprisingly sparse. 
Krishnan et al. \citep{krishnan2016privateclean} proposes differentially private data cleaning methods which support human-in-the-loop cleaning. The methods described enable an expert to specify rules for data cleaning, and ensure that the result of a query is differentially private, which may include the impact of the expert looking at data to generate the rules. As such, this is not comparable with our present direction/approach. \citep{clifton2022} is the only previous paper which specifically addresses the problem of imputation under differential privacy.  That work focused on a specific mechanism for differential privacy and used smooth sensitivity which requires a significant amount of work from a data curator.  Our work aims to derive more general results regarding imputation with privacy that alleviate some of those problems to aid in actual use. 

\section{Assumptions regarding the imputation scheme}

Let $D$ be a dataset, with values taken from some universe $\mathcal{U}$. Let $D_{obs}$ refer to the observed part of $D$, and $D_{mis}$ denote the subset of $D$ for which the information is missing. We let $n$ denote the number of units contained in the dataset and define $n_{mis}$ as the number of units for which at least one attribute is missing. 
By a \textbf{neighbor} of $D$ we mean a dataset $D'$ which differs from $D$ in exactly one tuple. 


We let $\iota: \mathcal{U}^n \to  \mathcal{U}^n$ 
denote the imputation scheme, the rule for assigning values to missing values. This means that 
$\forall D \in\mathcal{U}^n,\quad n_{mis}\left(\iota(D)\right) = 0.$

Throughout the paper we make the following assumptions about the imputation scheme $\iota$: 
\begin{enumerate}
    \item $\iota(D) \in \mathcal{U}^n$, that is, the imputed dataset is one which could have occurred without imputation. 
    \item The imputation scheme does not change the observed values. As a consequence, if $D$ and $D'$ are neighbors, then the Hamming distance between $\iota(D)$ and $\iota(D')$ is at most $n_{mis} + 1$. 
\end{enumerate}

\section{Addressing privacy concerns in the imputation context}

Given the two-step nature of analyzing imputed datasets--the imputation step always precedes the analysis step--there are two general approaches how privacy considerations can be taken into account. Privacy can either be addressed in both steps or only at the final stage when analyzing the data. (Only adjusting the imputation step is not an option as this would leave the originally observed values unprotected). 

\subsection{Addressing privacy only at the analysis stage}
It seems natural at first sight to only consider the privacy implications of imputation when computing the query of interest. Different queries might be affected differently and it seems natural to develop tailor made algorithms to avoid introducing noise for protection when it is not necessary. 

Unfortunately, it turns out that generally relying on such a strategy can result in high privacy costs. We illustrate this by first establishing an upper bound on the possible privacy degradation from imputation and then demonstrating that this upper bound is tight for some settings.

We can establish an upper bound by tying the missing data problem to the notion of group privacy. 
\begin{prop}[Group privacy]\label{prop:group}
If $D$ and $D'$ are datasets which differ in at most $k$ elements, and $\mathcal{M}$ is an $\varepsilon$-DP algorithm, then 
$Pr(\mathcal{M}(D) \in S) \leq Pr(\mathcal{M}(D')\in S) \times e^{k\varepsilon}.$

\end{prop}

Since under our assumptions of the imputation scheme $\iota$ two neighboring databases can differ in up two $n_{mis}+1$ elements after imputation, Proposition \ref{prop:group} tells us that for any dataset $D$ and any neighbor $D'$, 
\[
Pr(\mathcal{M}(\iota(D)) \in S) \leq Pr(\mathcal{M}(\iota(D'))\in S) \times e^{(n_{mis}+1)\varepsilon}.
\]
It's worth noting, however, that the bound depends on $D$ (through the term $n_{mis}$), so the uniform bound over all possible datasets of size $n$ would be $e^{n\varepsilon}$.

We can establish a similar bound for the global sensitivity. 
Let $q$ be a query, and $\Delta(q)$ denote its global sensitivity (defined on fully observed datasets), 
\[
\Delta(q) = \max_{\mathcal{Y} \in \mathcal{U}^n: \mathcal{D}_{mis} = \{\emptyset\}} \quad  \max_{D':d(D, D') = 1, D'_{mis} = \{\emptyset\}} |q(Y) - q(Y')|.
\]
Then we have the following 
\begin{prop}\label{Bounds}
\[
\Delta(q) \leq \max_{D \in \mathcal{U}^n, n_{mis} \geq 0} \quad  \max_{D':d(D, D') = 1} |q(\iota(D)) - q(\iota(D'))|
\] and for any $D \in \mathcal{U}^n$ with $n_{mis} \geq 0$,
\[
\max_{D':d(D, D') = 1} |q(\iota(D)) - q(\iota(D'))| \leq (n_{mis} +1) \Delta(q).
\]
\end{prop}
\begin{proof}
The first inequality is trivial, since the fully observed data is included in the space of datasets to be considered.

For the second inequality, consider any $D$ and $D'$ which are neighboring datasets. By the second condition on the imputation scheme, we have that $\iota(D)$ and $\iota(D')$ differ in at most $n_{mis} + 1$ entries. We can think of a chain of datasets, denoted $\{D_j\}^{n_{mis} + 1}_{j=0}$ where each dataset differs from the first by exactly one change, $D_0 = \iota(D)$, and $D_{n_{mis} + 1} = \iota(D')$. By repeated application of the triangle inequality, we have 
\begin{eqnarray*}
|q(\iota(D)) - q(\iota(D'))| \leq \sum^{n_{mis}}_{i=0} |q(D_i) - q(D_{i+1})| \leq (n_{mis} + 1) \Delta(q).
\end{eqnarray*}
In the last step we have used that since each dataset $D_j \in \mathcal{U}^n$ and $n_{mis} = 0$, the change observed is no larger than the global sensitivity of $q$ (this is a consequence of the first assumption on $\iota$ above). \end{proof}

In the case of using linear regression to impute a response variable $Y$ whose values are bounded by some range $[a,b]$, it is straightforward to construct an example of a dataset $D$ and a neighbor $D'$ for which 
\[
|q(\iota(D')) - q(\iota(D))| = \frac{a + (n-1)b}{n} - a = (n-1)\frac{b-a}{n} = \left(n_{mis} + 1\right) \Delta(q),
\]
showing that the upper bound proven in Proposition \ref{Bounds} is tight, at least without further assumptions on the imputation scheme and/or query.

\subsection{Addressing privacy concerns at the imputation stage}
To reduce the impacts of imputation on privacy in the worst case, it might be helpful to already account for privacy when imputing the missing values. In this section we illustrate that this strategy can indeed reduce the bounds in the worst case if we are willing to make two additional assumptions regarding the imputation scheme. These assumptions are fulfilled by most of the imputation schemes used in practice. 

The first requirement is that the imputed values of any record $i$, $i=1,\ldots,n$, are only a function of its observed values, that is,
\begin{equation}\label{eq:imp}
D_{imp}^{(i)}\sim m(D_{obs}^{(i)},\theta),
\end{equation}
where $m$ denotes the model and $\theta$ are the model parameters.

The second assumption is that the missingness mechanism is ignorable following \citep{rubin1976}. We do not provide the technical details for brevity, but informally, one key assumption of ignorable missingness mechanisms is that any systematic difference in the probability of a unit to be missing can be fully explained by the part of the data that is still observed. 

The important practical implications of an ignorable missingness mechanism is that correct inferences regarding the full data can be obtained without the need to specify the parameters of the missingness mechanism
. Furthermore, if we partition $D=\{Y,X\}$ where $Y$ contains those attributes that are only partially observed and $X$ contains those attributes that are fully observed it holds that
\[
f(Y_{obs}|X_{obs})=f(Y|X),
\]
where the index $obs$ refers to the $n-n_{mis}$ units that are fully observed. This is especially relevant in the imputation context as it implies that the fully observed cases can be used to estimate the parameters of $f(Y|X)$ and these parameters can then be used to impute any missing values. 

Since under non-ignorable missingness mechanisms, assumptions regarding the missing-data mechanism need to be established that can never be tested based on the observed information, most imputation models are based on the assumption that ignorability holds. Under this assumption, imputation is carried out in two steps: The parameters of the imputation model are estimated using the fully observed data. The parameters are then used to impute the missing values based on Equation (\ref{eq:imp}).

This two step procedure has important implications from the privacy perspective. If the parameters for the imputation model are estimated in a privacy preserving manner, the privacy guarantees no longer depend on the number of incomplete tuples $n_{mis}$. We can formalize this with the following theorem. 
\begin{theorem}
Let $\iota_{\varepsilon}(D)$ be an imputation scheme, which imputes missing values according to the model $D_{imp}^{(i)}\sim m(D_{obs}^{(i)},\hat{\theta})$, where $\hat{\theta}$ are the model parameters estimated using any suitable $\varepsilon_1$-differentially private mechanism. Given an $\varepsilon_2$-differentially private mechanism $\mathcal{M}$, we have that $\mathcal{M}(\iota_{\varepsilon}(D))$ is $(\varepsilon_1+\varepsilon_2)$ differentially private.
\end{theorem}
\begin{proof}
This statement follows from a general version of the Sequential Composition Theorem as stated in \citet[Theorem B.1]{DworkRoth}. This statement of the theorem allows the use of the output of the first mechanism, the model parameters $\hat\theta$, to inform the second mechanism. 
\end{proof}
The idea in the previous proof is that we can envision the imputation in the following way: when the survey is given out, we first invite any individuals willing to answer every question to submit their surveys. From these complete responses we learn the parameters $\theta$ for the distribution $f(Y|X;\theta)$. We then invite any survey respondent with an incomplete survey to use the model $m(D_{obs}^{(i)},\hat\theta)$ as in Equation (\ref{eq:imp}) above to fill in any blanks in their reply, and send it back to us. We can then run differentially private queries on this data complete dataset.  

\section{Experiments}

We are currently running experiments which aim at comparing the utility impacts of three different strategies to deal with missing values under DP: (i) available-case analysis (ii) addressing privacy at the analysis stage, and (iii) addressing privacy at both the imputation and analysis stages. Note that even though strategy (iii) implies strict bounds on privacy, whether this also results in higher accuracy of the query response for a given privacy budget largely depends on the accuracy of the algorithms used to obtain the model parameters $\theta$. 

In our experiments we assume for simplicity that missingness is limited to one variable and that the imputation model to be used is based on OLS regression. A literature review of previous approaches for OLS under DP identified the following candidates for estimating $\theta$: The functional mechanism \citep{zhang2012functional}, PrivGene \citep{zhang2013privgene}, another approach based on perturbing the objective function proposed by \citep{chaudhuri2011differentially}, an output perturbation approach proposed by \citep{wu2015revisiting}, and an approach based on robust statistics \citep{avella2021}. We only implemented the functional mechanism and non-differentially private imputation in our experiments so far, but hope to also implement the other approaches in the future. 

We use simulated data to have full control over the data generating process and the missingness mechanism. Specifically, we generate $X=\{X_1,X_2\}$ by independently drawing $n=10,000$ records from a uniform distribution bounded between 0 and 1. 
We generate $Y$ by drawing from $Y=X'\beta+\tau,\quad \tau\sim N(0,I_n\sigma^2),$
where the vector $\beta=\{0.5,0.5\}$ and $\sigma^2=0.1$. We use this model to ensure that the assumptions of the imputation model are satisfied, as we are only interested in assessing the effects of privacy considerations and not how a potential mis-specification of the imputation model affects the utility of the imputed data. We clip $Y$ to be within [0,1] to enforce bounds on $Y$. 
To introduce missing values, we use $Pr(M_Y=1)=X_0,$
where $M_Y$ is the missing data indicator for $Y$. 
In each iteration of the simulation we set $Y$ to missing according to the probabilities given by the model. We set $\varepsilon=1$ and split the privacy budget equally between the imputation and analysis step in strategy (iii). We assume that the query of interest is the mean of $Y$ and we use the Laplace mechanism to protect it.  
We note that this setup allows us to compare to a ground truth which helps to also measure the bias and not only the uncertainty introduced by the various approaches
.

Simulation results are shown in Figure  \ref{fig:experiments}. 
\begin{wrapfigure}{R}{0.5\linewidth}
\includegraphics[trim={0 1.2cm 0 1.2cm 0},width=\linewidth]{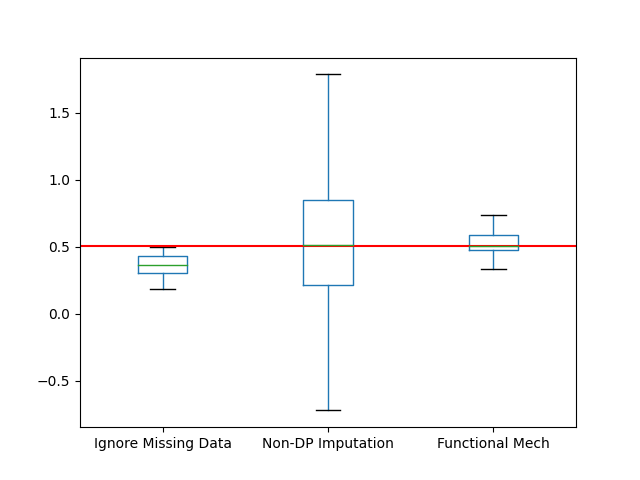}
\caption{ 
Boxplot for estimated means across simulations. Red line is true value. $\varepsilon = 1$.}
\label{fig:experiments}
\vspace{-1.5ex}
\end{wrapfigure}

The boxplots show the distribution of the estimated mean of $Y$ across 500 simulation runs for the different strategies to deal with the missing values. The red line indicates the true mean of $Y$ before introducing missing values. The results indicate that dropping the missing values introduces bias in the obtained estimates. Both imputation strategies provide unbiased results. However, the uncertainty introduced by strategy (iii) is substantially smaller than for strategy (ii).

In the future, we want to explore how changing various aspects of the simulation design (sample size, dimension of $X$, $\sigma^2$, $\varepsilon$, splitting of the budget for  strategy (iii), etc.)
affects the utility of the three approaches.

\bibliographystyle{plainnat}

\bibliography{Bibliography.bib}

\end{document}